\documentclass[conference]{IEEEtran}
\IEEEoverridecommandlockouts
\usepackage{cite}
\usepackage{amsmath,amssymb,amsfonts}
\usepackage{algorithmic}
\usepackage{graphicx}
\usepackage{textcomp}
\usepackage{xcolor}
\def\BibTeX{{\rm B\kern-.05em{\sc i\kern-.025em b}\kern-.08em
    T\kern-.1667em\lower.7ex\hbox{E}\kern-.125emX}}
\newtheorem{remark}{Remark}
\begin{document}

\title{Breakdown of Edgeworth Expansion\\ in Finite-Blocklength Regime\\ and Exact Absorption via $q$-Deformation\\
}

\author{\IEEEauthorblockN{Hiroki Suyari}
\IEEEauthorblockA{\textit{Graduate School of Informatics,} 
\textit{Chiba University,}
Chiba, Japan \\
suyari@faculty.chiba-u.jp, suyarilab@gmail.com}
}

\maketitle

\begin{abstract}
This paper addresses the structural breakdown of the Edgeworth expansion in the finite-blocklength (FBL) regime, where conventional asymptotic approximations yield unphysical negative probabilities in the deep-tail region. We propose a $q$-deformed framework that resolves this inconsistency by replacing additive polynomial perturbations with a geometric deformation of the information density space. Motivated by the linearization of nonlinear dynamics, we prove that dynamically scaling the $q$-logarithmic parameter exactly absorbs the third-order skewness while preserving global nonnegativity. We establish a universal asymptotic matching, demonstrating that the framework encapsulates higher-order asymptotic scales. Numerical results confirm that the proposed method matches the state-of-the-art precision of the Cornish-Fisher bound without the risk of negative probabilities. The framework offers a robust and computationally stable foundation for evaluating operational limits in ultra-reliable communications such as 6G and URLLC.
\end{abstract}

\section{Introduction}
\label{sec:introduction}

In finite-blocklength (FBL) information theory, pioneered by Polyanskiy, Poor, and Verdú \cite{Polyanskiy2010} and developed in subsequent literature (e.g., \cite{Tan2014}), the precise evaluation of decoding error probabilities at short blocklengths is required. To characterize the fundamental limits beyond the asymptotic capacity, the distribution of the information density must be analyzed.

This finite-size analysis is essential for the design of modern ultra-reliable low-latency communications (URLLC) envisioned for 6G networks and IoT systems \cite{Shirvanimoghaddam2019}. In such applications, strict latency constraints dictate the use of extremely short packets (often $n \sim 100$ to $200$ channel uses), while simultaneously demanding ultra-low target error probabilities (e.g., $\epsilon \sim 10^{-5}$ to $10^{-9}$). Consequently, the operational point of these systems lies deep within the tail of the information density distribution, where the conventional Gaussian approximation (Central Limit Theorem) severely underestimates the physical communication penalty.

Conventionally, the finite-size scaling of the error probability is evaluated using the normal approximation. To account for the asymmetry inherent in practical channels or sources, higher-order asymptotic expansions---most notably the Edgeworth expansion and its inverted form, the Cornish-Fisher expansion---are standardly employed to refine the tail bounds \cite{Feller1971}. These perturbative approaches append additive orthogonal polynomials (e.g., Hermite polynomials) to the Gaussian baseline to correct for skewness and kurtosis.

However, while additive polynomial corrections provide accurate local approximations near the bulk of the distribution, they structurally break down in the deep-tail regime (i.e., rare events associated with ultra-low error probabilities). Because polynomials inevitably exhibit oscillatory behavior and roots, truncating the Edgeworth expansion at a finite order generates unphysical zero-crossing regions, yielding strictly negative probabilities.

In this paper, we propose a structural resolution to this mathematical inconsistency. Rather than relying on additive polynomial perturbations, we demonstrate that finite-size fluctuations can be exactly and globally absorbed by algebraically deforming the information measure itself. By introducing a dynamically scaled $q$-deformed framework---originating from generalized statistical mechanics \cite{Tsallis1988,Tsallis2009}---we prove that the non-Gaussian skewness is absorbed into an exponential structure, thereby strictly preserving the nonnegativity of the probability measure while exactly matching the established FBL asymptotic bounds.

\section{Breakdown of Conventional Edgeworth Expansion}
\label{sec:breakdown}

Let $X^n = (X_1, X_2, \dots, X_n)$ be a sequence of length $n$ generated by a discrete memoryless source (DMS) $P_X$. The fundamental random variable governing the finite-blocklength performance is the information density, defined as $i(X^n) = \sum_{k=1}^n -\ln P_X(X_k)$. 

To rigorously define the macroscopic fluctuations, let $H = \mathbb{E}[-\ln P_X(X)]$ denote the single-letter entropy, and $V = \mathrm{Var}[-\ln P_X(X)]$ denote the single-letter varentropy. Because the source is memoryless, the expectation of the information density over the entire sequence is the \textit{block entropy} $nH$, and its variance is the \textit{block varentropy} $nV$. Furthermore, let $T = \mathbb{E}[(-\ln P_X(X) - H)^3]$ be the single-letter third central moment.

To evaluate the distribution of the information density, it is conventional to define the standardized macroscopic fluctuation $Z = (i(X^n) - nH)/\sqrt{nV}$. The standard third-order Edgeworth expansion for its probability density function (PDF) $f_Z(z)$ is mathematically expressed as:
\begin{equation}
    f_Z(z) = \phi(z) \left[ 1 + \frac{\gamma_1}{6\sqrt{n}} (z^3 - 3z) \right] + O(n^{-1}),
    \label{eq:edgeworth_pdf}
\end{equation}
where $\phi(z) = \frac{1}{\sqrt{2\pi}}e^{-z^2/2}$ is the standard normal PDF, and $\gamma_1 = T/V^{3/2}$ represents the skewness of the single-letter information density.

While \eqref{eq:edgeworth_pdf} provides a mathematically rigorous asymptotic expansion as $n \to \infty$, its truncation at $O(n^{-1})$ exhibits a structural inconsistency in the deep-tail regime for finite $n$, particularly in the large deviation regime.

To make this explicit, let us configure a highly asymmetric binary memoryless source (BMS) with symbol probabilities $P_X(0) = 0.1$ and $P_X(1) = 0.9$. 
For a short blocklength of $n=12$, the skewness of the information density evaluates to a strictly positive value ($\gamma_1 / \sqrt{n} \simeq 0.77$). 

Substituting these parameters into \eqref{eq:edgeworth_pdf} reveals that in the deep-tail region where $z \lesssim -2.45$, the polynomial correction term $\frac{\gamma_1}{6\sqrt{n}} (z^3 - 3z)$ evaluates to a value less than $-1$. 
Consequently, the expansion yields a strictly negative probability density, $f_Z(z) < 0$, which mathematically violates the nonnegativity axiom of probability. Furthermore, attempting to resolve this by appending the fourth-order kurtosis correction only amplifies the violation of nonnegativity, leading to a deeper negative region. As illustrated in Fig.~\ref{fig:breakdown}, this zero-crossing invalidates the additive perturbative approach for deep-tail finite-blocklength analysis. To resolve this, rather than appending additive polynomials, we propose absorbing this finite-size skewness structurally via the $q$-deformed algebraic framework of Tsallis statistics.

\begin{figure}[htbp]
    \centerline{\includegraphics[width=\columnwidth]{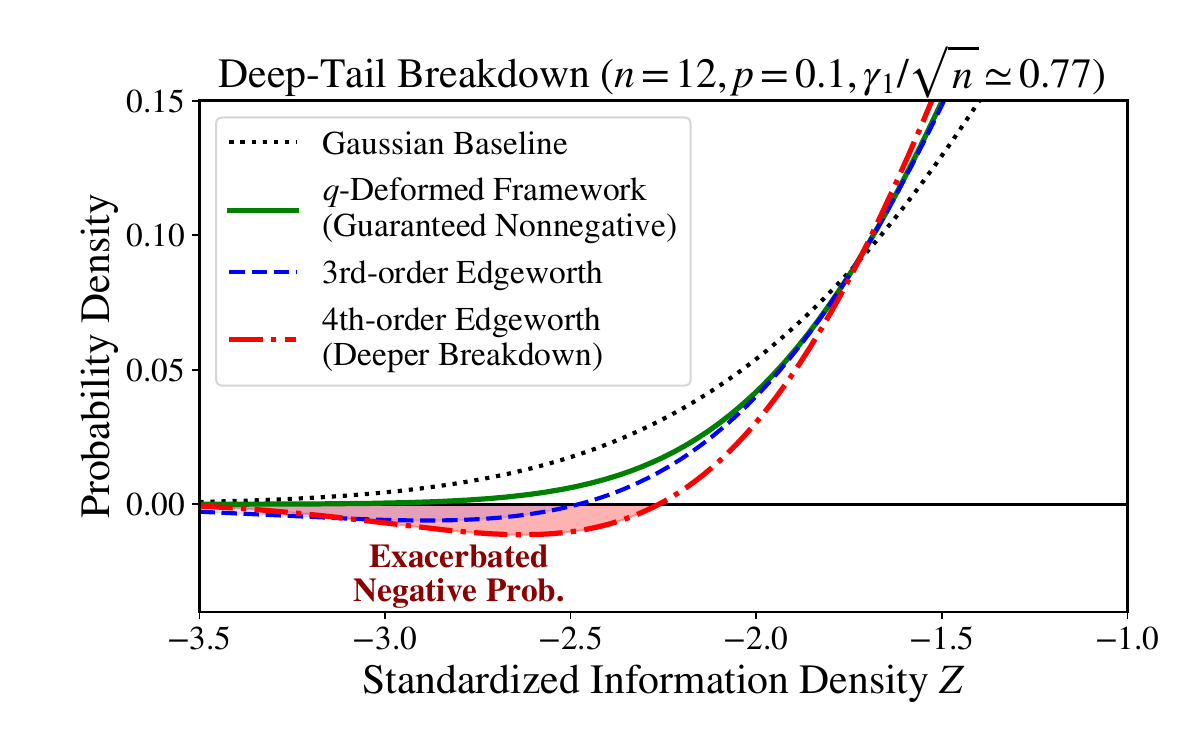}}
    \caption{Breakdown of the conventional Edgeworth expansion in the deep-tail regime for a binary memoryless source ($n=12, P_X(0)=0.1$). The third-order Edgeworth PDF crosses zero around $z \simeq -2.45$. Appending the fourth-order correction further exacerbates this unphysical negative probability. In contrast, the proposed $q$-deformed framework structurally preserves nonnegativity across the entire domain.}
    \label{fig:breakdown}
\end{figure}

For a given tail probability $\epsilon \in (0,1)$, the third-order asymptotic threshold (quantile) of the information density, $L_{\text{edge}}$, is rigorously evaluated via the Cornish-Fisher expansion \cite{Tan2014, Feller1971}:
\begin{equation}
    L_{\text{edge}} = nH + \sqrt{nV} Z_{\epsilon} + \frac{T}{6V} (Z_{\epsilon}^2 - 1) + O(n^{-1/2}),
    \label{eq:edgeworth_expansion}
\end{equation}
where $Z_\epsilon$ is the $(1-\epsilon)$-quantile of the standard normal distribution. The $O(1)$ term $\frac{T}{6V} (Z_{\epsilon}^2 - 1)$ represents the finite-size correction for skewness. Our primary goal in this paper is the exact structural absorption of this $O(1)$ skewness penalty via the $q$-deformed framework, thereby guaranteeing strict probability nonnegativity.

\section{Exact Structural Absorption via $q$-deformed framework}
\label{sec:q_algebra}

Before detailing the exact algebraic formulation, we clarify the theoretical motivation for introducing the $q$-logarithm into finite-blocklength (FBL) information theory.

In classical asymptotic theory ($n \to \infty$), the standard information density adheres to the conventional logarithm, $-\ln P(X^n)$. From a structural perspective, adhering to the standard logarithm implicitly adopts the Gaussian measure as the baseline reference, owing to the strict additivity and its associated Central Limit Theorem. Consequently, any attempt to incorporate higher-order FBL penalties (such as skewness) within this standard logarithmic framework must rely on multiplicative polynomial perturbations applied to the probability density (i.e., the Edgeworth expansion). As demonstrated in Section II, this polynomial-based correction is incompatible with the global tail structure, leading to a violation of the nonnegativity axiom in the deep-tail regime.

To resolve this structural inconsistency, the correction must be a geometric deformation of the information measure itself. The mathematical necessity of the $q$-logarithm is best understood through its differential characterization. The standard logarithm uniquely linearizes the simplest linear dynamics $dy/dx = y$, yielding $d(\ln y)/dx = 1$, which underpins classical additivity. When finite-size fluctuations introduce nonlinear scaling into the probability measure, the canonical structural deformation of this dynamic is $dy/dx = y^q$. This separable differential equation is uniquely linearized by the $q$-logarithm, yielding a constant derivative $d(\ln_q y)/dx = 1$, which systematically restores the additive structure in the deformed information space. 

Thus, the $q$-logarithm is not an arbitrary curve-fitting function, but the unique canonical operator that mathematically linearizes the simplest nonlinear dynamics. Formally, it is defined for $q \neq 1$ as $\ln_q x := \frac{x^{1-q}-1}{1-q}$, which smoothly recovers the standard natural logarithm $\ln x$ in the limit $q \to 1$. Originally developed in generalized statistical mechanics \cite{Tsallis2009}, this algebraic approach is deeply connected to the generalized law of error \cite{Suyari2005}, providing a mathematically rigorous foundation for handling non-Gaussian anomalous fluctuations without violating probability axioms.

By embedding this $q$-deformed structure into the density space and treating $q$ as a dynamic parameter dependent on the blocklength $n$, we effectively absorb the FBL penalty. Let $W_n := i(X^n) - nH$ be the centered conventional information density, which scales as $O(n^{1/2})$. Instead of appending perturbative polynomials to the distribution, we apply the $q$-logarithmic transformation to the exponential fluctuation, yielding $\ln_q(\exp(W_n)) = \frac{\exp((1-q)W_n)-1}{1-q}$. To strictly preserve the exact mean $\mathbb{E}[i_{q_n}(X^n)] = nH$, we center this deformation via the moment-generating function (MGF) of $W_n$. Thus, with a dynamic scaling parameter $q_n$, we define the centralized $q$-information density as:
\begin{equation}
    i_{q_n}(X^n) := nH + \frac{\exp((1-q_n)W_n) - \mathbb{E}[\exp((1-q_n)W_n)]}{1-q_n}.
    \label{eq:q_info_density}
\end{equation}
By expanding \eqref{eq:q_info_density} algebraically, the dominant fluctuation is embedded into the exponential structure. We formally state the exact absorption of the finite-size penalty as follows:

\newtheorem{theorem}{Theorem}
\begin{theorem}[Exact Structural Absorption]
\label{thm:absorption}
Let $L_q$ be the threshold of the $q$-generalized measure satisfying the probability bound for a target error probability $\epsilon \in (0, 1)$:
\begin{equation}
    P\left( i_{q_n}(X^n) \le L_q \right) \ge 1-\epsilon.
\label{eq:q_fluctuation_bound}
\end{equation}
By setting the dynamic scaling law of the deformation parameter as $1-q_n = \alpha n^{-1}$ with the tuning constant:
\begin{equation}
    \alpha = \frac{T}{3V^2},
\label{eq:alpha_tuning}
\end{equation}
where $V$ is the single-letter varentropy and $T$ is the single-letter third central moment, the threshold $L_q$ asymptotically coincides with the third-order finite-size boundary $L_{\text{edge}}$ derived from the Cornish-Fisher expansion up to $O(n^{-1/2})$.
\end{theorem}

\begin{IEEEproof}
To evaluate the distribution of the $q$-information density $i_{q_n}(X^n)$, we apply a Taylor expansion to the definition in \eqref{eq:q_info_density} around $q_n = 1$. Let $\delta_n := 1-q_n$. Expanding the exponential function $\exp(\delta_n W_n)$ up to the second order yields:
\begin{align}
    i_{q_n}(X^n) &= nH + \frac{1}{\delta_n} \left[ \left(1 + \delta_n W_n + \frac{\delta_n^2}{2}W_n^2 + \dots \right) \right. \notag \\
    &\quad \left. - \mathbb{E}\left[ 1 + \delta_n W_n + \frac{\delta_n^2}{2}W_n^2 + \dots \right] \right].
\end{align}
Since the centered macroscopic fluctuation $W_n = i(X^n) - nH$ has a zero mean ($\mathbb{E}[W_n] = 0$) and its variance is the block varentropy ($\mathbb{E}[W_n^2] = nV$), the centralized expansion simplifies to:
\begin{equation}
    i_{q_n}(X^n) = nH + W_n + \frac{\delta_n}{2}(W_n^2 - nV) + O(\delta_n^2).
    \label{eq:iq_expansion_Wn}
\end{equation}
To analyze this in terms of the standardized macroscopic fluctuation, we substitute $W_n = \sqrt{nV}Z$. Let $Z_\epsilon$ be the $(1-\epsilon)$-quantile of the standard normal distribution $\mathcal{N}(0,1)$, which corresponds to the standardized threshold for the target error probability $\epsilon$. Evaluating the $q$-generalized random variable at this specific quantile $Z_\epsilon$ gives the boundary $L_q$:
\begin{align}
    L_q &= nH + \sqrt{nV}Z_\epsilon + \frac{\delta_n}{2} \left( (\sqrt{nV}Z_\epsilon)^2 - nV \right) + O(n^{-1/2}) \notag \\
    &= nH + \sqrt{nV}Z_\epsilon + \frac{\delta_n nV}{2}(Z_\epsilon^2 - 1) + O(n^{-1/2}).
    \label{eq:Lq_derivation}
\end{align}
In classical probability theory, the third-order threshold bounded by the standard information density is rigorously given by the Cornish-Fisher expansion $L_{\text{edge}}$ in \eqref{eq:edgeworth_expansion}. To structurally absorb the finite-size skewness without introducing additive polynomials, we enforce the mathematical equivalence $L_q = L_{\text{edge}}$ up to the $O(1)$ constant term. Equating the $O(1)$ correction term in \eqref{eq:Lq_derivation} with the corresponding term $\frac{T}{6V}(Z_\epsilon^2 - 1)$ in \eqref{eq:edgeworth_expansion} yields:
\begin{equation}
    \frac{\delta_n nV}{2}(Z_\epsilon^2 - 1) = \frac{T}{6V}(Z_\epsilon^2 - 1).
\end{equation}
Since this geometric equivalence must hold for any target tail probability $\epsilon$ (i.e., for any quantile $Z_\epsilon$), we equate the coefficients:
\begin{equation}
    \frac{(1-q_n)nV}{2} = \frac{T}{6V} \implies 1-q_n = \frac{T}{3V^2}n^{-1}.
\end{equation}
This establishes the exact dynamic scaling law $1-q_n = \alpha n^{-1}$ with the specific algebraic tuning $\alpha = T/(3V^2)$, structurally absorbing the third-order skewness penalty into the $q$-exponential measure.
\end{IEEEproof}

While the conventional expansion (Fig.~\ref{fig:breakdown}) additively appends the skewness $\gamma_1$, leading to a polynomial zero-crossing, the $q$-deformed framework absorbs this skewness penalty directly into the parameter $\alpha$ within the generalized exponential structure. Consequently, Equation \eqref{eq:q_info_density} structurally guarantees the strict nonnegativity of the probability measure across the entire domain, rendering it fully valid for deep-tail finite-blocklength analysis.

\begin{remark}
While Theorem 1 establishes the exact structural absorption up to the third-order moment (skewness), an extension to strictly absorb the fourth-order moment (kurtosis) is possible. This would require a higher-order dynamic scaling of the deformation parameter, which is beyond the scope of this paper and remains an important subject of our ongoing work.
\end{remark}

\section{Asymptotic Matching at Higher Orders}
\label{sec:higher_order}

One might intuitively argue that the deep-tail inconsistency of the third-order Edgeworth expansion could be resolved by evaluating higher-order perturbative terms (e.g., the fourth-order kurtosis correction). However, this is incorrect from an asymptotic viewpoint.
Because the Edgeworth series is an asymptotic expansion rather than a convergent power series, its higher-order terms are governed by higher-degree Hermite polynomials. For a fixed, finite blocklength $n$, appending higher-degree polynomials introduces additional roots, which exacerbates the oscillatory behavior in the deep-tail regime. As visually demonstrated earlier in Fig.~\ref{fig:breakdown}, higher-order approximations generate even more severe zero-crossing regions, demonstrating that polynomial perturbation is structurally inadequate for large deviation evaluation.

By contrast, the proposed $q$-deformed framework universally encapsulates higher-order fluctuations without relying on additive polynomials. This structural correspondence is formally stated as follows.

\begin{theorem}[Universal Asymptotic Matching]
\label{thm:resonance}
Let the centered fluctuation be $W_n = O(n^{1/2})$ and the generalized parameter scale as $1-q_n = O(n^{-1})$. The $k$-th degree term of the $q$-deformed expansion exactly matches the asymptotic order of the $(k+1)$-th moment correction in the classical Edgeworth expansion, yielding the universal order $O(n^{1-k/2})$.
\end{theorem}

\begin{IEEEproof}
By expanding the centralized $q$-information density \eqref{eq:q_info_density} via the Taylor series of the exponential function, we obtain:
\begin{equation}
    i_{q_n}(X^n) = nH + W_n + \sum_{k=2}^{\infty} \frac{(1-q_n)^{k-1}}{k!} \left( W_n^k - \mathbb{E}[W_n^k] \right).
    \label{eq:general_higher_order_expansion}
\end{equation}
By substituting the standard growth rate of the fluctuation $W_n = O(n^{1/2})$ and the dynamic scaling limit $1-q_n = O(n^{-1})$ into the $k$-th degree term, the asymptotic order rigorously evaluates to:
\begin{equation}
    O\left( (n^{-1})^{k-1} \right) \times O\left( (n^{1/2})^k \right) = O\left( n^{1-k/2} \right).
    \label{eq:general_order_resonance}
\end{equation}
Evaluating \eqref{eq:general_order_resonance} for successive values of $k$ reveals a systematic structural alignment with the classical finite-blocklength limits:
\begin{itemize}
    \item $k=1 \implies O(n^{1/2})$: Recovers the normal fluctuation (varentropy penalty).
    \item $k=2 \implies O(1)$: Generates the third-order skewness correction, which is exactly absorbed via $\alpha = T/(3V^2)$ as proven in Theorem \ref{thm:absorption}.
    \item $k=3 \implies O(n^{-1/2})$: Matches the asymptotic order of the fourth-order kurtosis penalty.
    \item $k=m \implies O(n^{1-m/2})$: Aligns with the $(m+1)$-th moment correction order for any integer $m \ge 1$.
\end{itemize}
\end{IEEEproof}

The generalized order relation $O(n^{1-k/2})$ indicates that the dynamic parameter $q_n$ acts as an exact algebraic generator for finite-length corrections. The single scaling rule $1-q_n = O(n^{-1})$ establishes a structural correspondence that systematically offsets the anomalous behavior of higher-order moments, entirely obviating the combinatorial complexity and structural instability of accumulating orthogonal polynomials.

While this paper has rigorously proven the exact absorption of the non-Gaussian fluctuations up to the third order (skewness), the explicit emergence of the $O(n^{-1/2})$ scale for $k=3$ suggests that this generalized framework structurally encapsulates the kurtosis and beyond. 

\section{Information-Theoretic Consequences of Truncation}
\label{sec:kl_divergence}

As established in Section \ref{sec:breakdown}, the finite-order Edgeworth expansion inevitably yields unphysical negative probability densities in the deep-tail large deviation regime. In engineering practice, a common ad-hoc workaround is to apply a "truncate-and-renormalize" procedure: negative values are artificially set to zero, and the remaining distribution is rescaled to ensure a valid probability measure. While this might seem like a trivial fix for numerical evaluations, we demonstrate that it structurally destroys the information-theoretic consistency of the model.

To formalize this, let $P_Z(z)$ be the true continuous probability density function of the standardized information density. For a continuous memoryless source, $P_Z(z)$ remains strictly positive over its entire support. Let $Q_{\text{trunc}}(z)$ denote the distribution approximated by the truncated Edgeworth expansion, defined as:
\begin{equation}
    Q_{\text{trunc}}(z) = \frac{1}{\mathcal{Z}} \max\left\{0, \phi(z) \left[ 1 + \frac{\gamma_1}{6\sqrt{n}} (z^3 - 3z) \right] \right\},
\end{equation}
where $\mathcal{Z}$ is the normalization constant required to ensure $\int Q_{\text{trunc}}(z)dz = 1$.

Because the uncorrected polynomial perturbation crosses zero, there necessarily exists a deep-tail region $\mathcal{R}$ where $Q_{\text{trunc}}(z) = 0$, whereas the true distribution maintains $P_Z(z) > 0$. In information theory, the fundamental distance measure between two distributions is the Kullback-Leibler (KL) divergence (or relative entropy), which dictates the fundamental limits of source coding and channel coding (e.g., via Sanov's theorem). Evaluating the KL divergence from the true distribution $P_Z$ to the truncated approximation $Q_{\text{trunc}}$ yields:
\begin{align}
    &D(P_Z \parallel Q_{\text{trunc}}) = \int_{-\infty}^{\infty} P_Z(z) \ln \frac{P_Z(z)}{Q_{\text{trunc}}(z)} dz \notag \\
    &= \int_{\mathbb{R} \setminus \mathcal{R}} P_Z(z) \ln \frac{P_Z(z)}{Q_{\text{trunc}}(z)} dz + \int_{\mathcal{R}} P_Z(z) \ln \frac{P_Z(z)}{0} dz \notag \\
    &\to \infty.
    \label{eq:kl_divergence_infinity}
\end{align}

The divergence to infinity in \eqref{eq:kl_divergence_infinity} implies a severe breakdown of the truncated model. It demonstrates that substituting $Q_{\text{trunc}}$ for the true distribution in any rigorous information-theoretic error exponent analysis or capacity bound derivation is mathematically invalid, as the fundamental distance measure completely breaks down.

In contrast, the proposed $q$-deformed framework algebraically absorbs the finite-size skewness without introducing polynomial roots. As detailed in Section \ref{sec:q_algebra}, the $q$-generalized probability density $Q_q(z)$ derived from the centralized $q$-information density intrinsically guarantees nonnegativity across the relevant domain by dynamically scaling $1-q_n = O(n^{-1})$. Consequently, the support of $Q_q(z)$ fully covers that of $P_Z(z)$, ensuring that the KL divergence remains strictly finite:
\begin{equation}
    D(P_Z \parallel Q_q) < \infty.
\end{equation}

This global stability provides a definitive practical benefit over trivial non-negative fixes. The $q$-deformed framework establishes a structurally consistent probabilistic model that prevents the mathematical breakdown of error-probability estimations, offering a robust foundation for evaluating operational limits in ultra-reliable finite-blocklength communications.

\section{Application to Finite-Blocklength Coding Rate}
\label{sec:coding_rate}

The ultimate objective of FBL information theory is to evaluate the maximal achievable coding rate $R^*(n, \epsilon)$ for a given blocklength $n$ and error probability $\epsilon$. The standard asymptotic expansion pioneered by Polyanskiy et al. \cite{Polyanskiy2010} approximates this rate as:
\begin{equation}
    R^*(n, \epsilon) = C + \sqrt{\frac{V}{n}} Z_\epsilon + O\left(\frac{\log n}{n}\right),
    \label{eq:polyanskiy_rate}
\end{equation}
where $C$ is the channel capacity (or source entropy). To achieve higher precision for asymmetric systems, the $O(1/n)$ skewness correction derived from the Cornish-Fisher expansion is conventionally appended:
\begin{equation}
    R_{\text{edge}}^*(n, \epsilon) = C + \sqrt{\frac{V}{n}} Z_\epsilon + \frac{T}{6V n}(Z_\epsilon^2 - 1).
    \label{eq:rate_with_skewness}
\end{equation}

However, as demonstrated in Section \ref{sec:breakdown}, translating the additive polynomial boundary $L_{\text{edge}}$ directly into the rate formulation inherits the zero-crossing issue in the ultra-low error regime (i.e., deep tail). 

By contrast, substituting the $q$-generalized threshold $L_q$ into the rate derivation yields the $q$-deformed achievable rate:
\begin{equation}
    R_q^*(n, \epsilon) = C + \sqrt{\frac{V}{n}} Z_\epsilon + \frac{1-q_n}{2} V (Z_\epsilon^2 - 1).
    \label{eq:q_rate}
\end{equation}
When the scaling law $1-q_n = \frac{T}{3V^2} n^{-1}$ is applied, \eqref{eq:q_rate} exactly recovers the asymptotic scaling of \eqref{eq:rate_with_skewness} up to the $O(n^{-1})$ penalty term. More importantly, because $R_q^*$ is structurally derived from the centralized $q$-information density $i_{q_n}$---which strictly guarantees the nonnegativity of the probability density across the entire domain by algebraically absorbing the finite-size penalty, unlike truncated polynomials---it provides a mathematically consistent formulation for FBL network designs (e.g., ultra-reliable low-latency communications in 6G) where conventional additive approximations structurally fail.

The practical impact of this exact structural absorption is visualized in Fig.~\ref{fig:coding_rate}, which compares the achievable finite-blocklength source coding rates for a target error probability $\epsilon = 10^{-3}$. The exact finite-length limit (staircase) demonstrates the discrete nature of the BMS. While the standard Gaussian approximation strictly deviates from the true bound, both the third-order Cornish-Fisher expansion and the proposed $q$-deformed framework precisely capture the finite-size correction, achieving high precision. 

However, the critical distinction lies in their structural foundation. As demonstrated in Section \ref{sec:numerical}, the Cornish-Fisher rate \eqref{eq:rate_with_skewness} achieves this tightness by implicitly integrating over unphysical negative probabilities in the deep tail. Conversely, the $q$-generalized rate \eqref{eq:q_rate} delivers the exact same state-of-the-art precision while being rigorously derived from a globally consistent, nonnegative probability measure.

The numerical agreement between $R_{\text{edge}}^*$ and $R_q^*$ in the third-order regime might intuitively suggest that the conventional approach is sufficient. However, this relies on evaluating the rate solely at a specific quantile threshold, masking the structural breakdown. As demonstrated in Section \ref{sec:numerical}, the Cornish-Fisher rate achieves this specific value by implicitly yielding unphysical negative probabilities in the deep-tail domain. 

Furthermore, extending the conventional threshold formulation $L_{\text{edge}}$ to higher orders requires inverting higher-degree Edgeworth polynomials, leading to a combinatorial explosion of mixed moments that is analytically intractable. Conversely, the $q$-generalized rate \eqref{eq:q_rate} delivers the exact same state-of-the-art precision while being rigorously derived from a globally consistent, nonnegative probability measure. Supported by the universal asymptotic matching (Theorem \ref{thm:resonance}), the proposed framework provides a mathematically sound and highly extensible foundation for evaluating operational limits in ultra-reliable FBL systems, where conventional additive approximations structurally and analytically fail.

\begin{figure}[tbp]
    \centerline{\includegraphics[width=0.90\columnwidth]{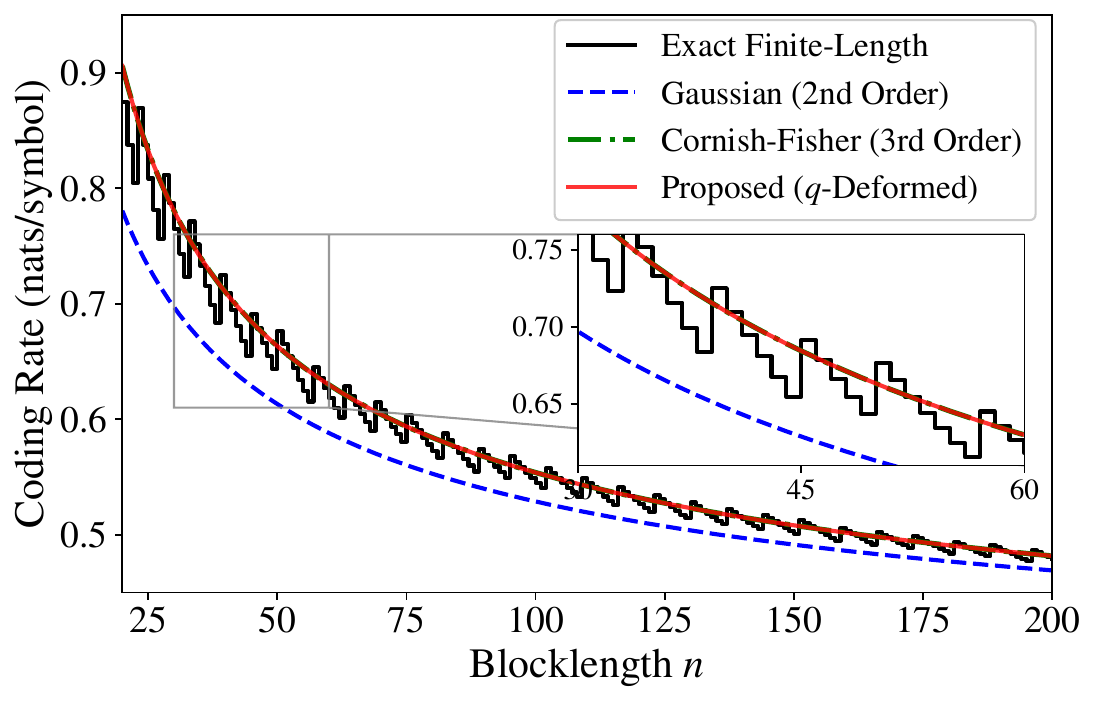}}
    \caption{Achievable finite-blocklength coding rates for the asymmetric BMS ($P_X(0)=0.1, \epsilon=10^{-3}$). The proposed $q$-deformed bound (red solid line) mathematically coincides with the highly accurate third-order Cornish-Fisher bound (green dashed line), precisely matching the exact discrete limit. It achieves this precision while structurally guaranteeing probability nonnegativity, overcoming the structural limitation of polynomial approximations.}
    \label{fig:coding_rate}
\end{figure}

\section{Numerical Case Study: Global Consistency}
\label{sec:numerical}

To concretely illustrate the structural absorption mechanism, we revisit the highly asymmetric binary memoryless source (BMS) introduced in Section \ref{sec:breakdown}. For this source ($P_X(0)=0.1$), the fundamental information-theoretic moments (in nats) are $H \simeq 0.325$, $V \simeq 0.434$, and $T \simeq 0.764$.

By applying Theorem \ref{thm:absorption}, the dynamic tuning parameter $\alpha$ is exactly evaluated as:
\begin{equation}
    \alpha = \frac{T}{3V^2} = \frac{0.764}{3 \times (0.434)^2} \simeq 1.35.
\end{equation}
Enforcing the scaling law $1-q_n = 1.35 n^{-1}$ completely embeds the third-order non-Gaussian fluctuation into the $q$-exponential family.

While Fig.~\ref{fig:breakdown} highlighted the local structural breakdown of the Edgeworth expansion in the extreme deep-tail regime, Fig.~\ref{fig:comparison} demonstrates the global behavior of the probability distributions. As observed in the main plot, all approximations successfully coincide in the bulk of the distribution (near $z \simeq 0$). 

\begin{figure}[htbp]
    \centerline{\includegraphics[width=\columnwidth]{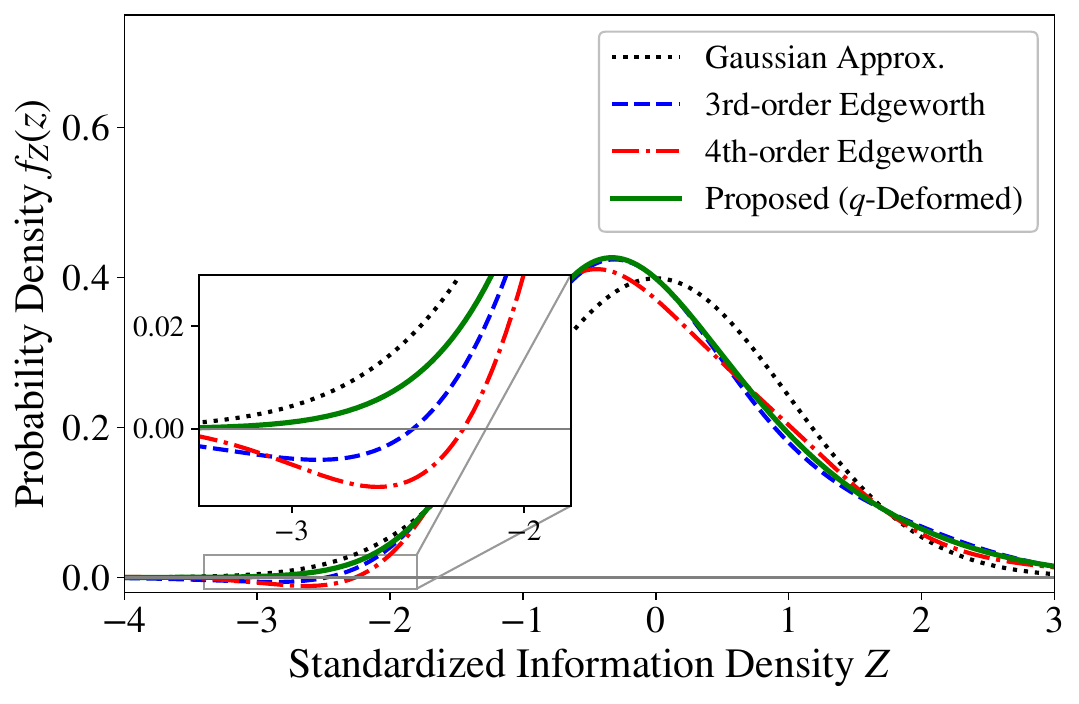}}
    \caption{Probability density functions for the asymmetric BMS ($n=12, P_X(0)=0.1$). Main plot: The overall distribution. Inset: Magnified view of the deep-tail regime where the conventional Edgeworth expansions (3rd and 4th order) become unphysically negative, while the proposed $q$-deformed framework exactly absorbs the skewness and remains nonnegative.}
    \label{fig:comparison}
\end{figure}

However, as highlighted by the inset, only the proposed $q$-deformed framework strictly maintains nonnegativity while accommodating the positive skewness ($\gamma_1/\sqrt{n} \simeq 0.77$ for $n=12$). The conventional cubic correction $\frac{\gamma_1}{6\sqrt{n}} (z^3 - 3z)$ forces the probability density to cross zero at $z \simeq -2.45$, and the fourth-order correction explicitly deepens this unphysical trough. In contrast, our framework stretches the probability tail geometrically. This confirms that a single algebraic operation structurally resolves the zero-crossing issue globally, establishing a theoretically consistent and computationally stable model.

\section{Discussion: Operational Tightness and Conservative System Guarantee}
\label{sec:discussion}

The proposed $q$-deformed framework resolves the nonnegativity breakdown of the Edgeworth expansion. Beyond resolving this mathematical inconsistency, this framework provides critical theoretical and practical advantages for finite-blocklength system designs.

\subsection{Unsafe Underestimation vs. Conservative Upper Bound}
In information theory, accurately estimating the deep-tail error probability is essential for URLLC systems. When the conventional Edgeworth expansion yields negative probabilities, an ad-hoc "truncate-and-renormalize" workaround is often applied. However, truncating the negative density implies that the probability of extremely rare events is artificially evaluated as strictly zero. In practical network design, this results in an \textit{unsafe underestimation}: a system might be falsely certified to meet ultra-reliable targets (e.g., $\epsilon = 10^{-9}$), leading to severe reliability degradation in practical deployments.

By contrast, the $q$-deformed framework algebraically absorbs the skewness to strictly preserve the probability measure's nonnegativity. As a mathematical consequence of preventing polynomial zero-crossings and maintaining a valid measure, the proposed model slightly overestimates the error probability in the extreme deep-tail regime compared to the exact limit. In the context of performance bounds (e.g., achievability bounds), providing an overestimation serves as a \textit{conservative upper bound}. Designing a communication system based on this conservative bound is theoretically sound, as it guarantees that the physical system's true reliability will strictly exceed the modeled estimates.

\subsection{Operational Tightness in Practice}
One might argue that a conservative bound is only useful if it does not overly penalize the system design. We emphasize that the $q$-deformed framework is not a loose bound; rather, it maintains operational tightness. As proven in Theorem 2, the dynamic scaling $1-q_n = O(n^{-1})$ ensures that the framework exactly matches the $O(n^{-1/2})$ finite-size asymptotic behavior of the source.

\begin{figure}[htbp]
    \centerline{\includegraphics[width=\columnwidth]{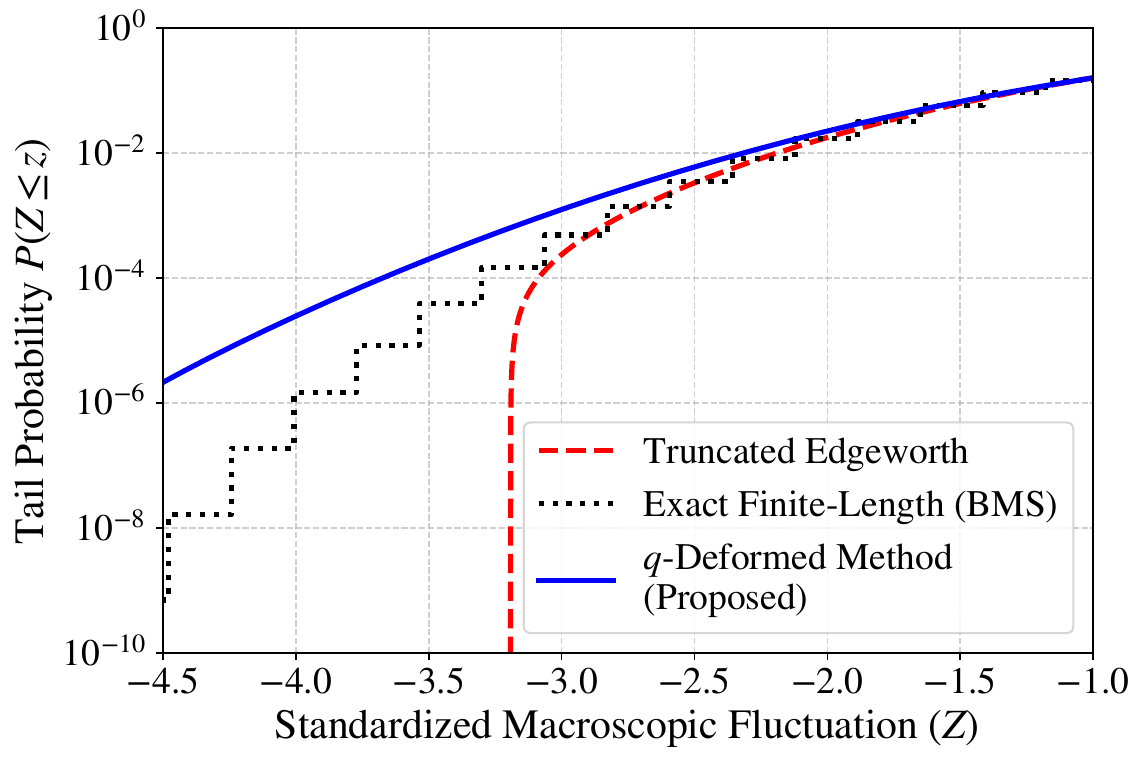}}
    \caption{Logarithmic tail probability $P(Z \le z)$ in the deep-tail regime for the asymmetric binary memoryless source (BMS) with $n=200$ and $P_X(0)=0.1$. The exact limit (dotted black line) exhibits a staircase structure reflecting the discrete nature of the source. The truncated Edgeworth expansion (dashed red line) drops to zero (rendering as a vertical drop in log scale), severely underestimating the error probability. The proposed $q$-deformed method (solid blue line) maintains nonnegativity, providing a mathematically safe, conservative upper bound that remains operationally tight to the discrete limit.}
    \label{fig:log_tail_probability}
\end{figure}

To quantify this, Fig.~\ref{fig:log_tail_probability} illustrates the logarithmic tail probability for the asymmetric BMS introduced in Section VII. Here, we evaluate a longer blocklength of $n=200$ to ensure the physical support of the discrete source extends sufficiently deep into the large deviation regime. The exact limit is represented by a step function due to the discrete nature of the source. The truncated Edgeworth expansion exhibits a sharp vertical drop where the estimated probability vanishes, highlighting its limitation for reliability estimation. Conversely, the $q$-deformed probability closely traces the exact staircase limit. While it provides a conservative overestimation to preserve the probability axioms, the gap remains sufficiently small within practical operating regimes (e.g., down to $10^{-6}$), ensuring that system resources (such as blocklength or power) are not unnecessarily wasted.

\subsection{Algorithmic Stability for 6G Implementations}
Furthermore, evaluating conventional bounds with higher-degree Hermite polynomials involves combinatorial complexity and severe numerical oscillation. The $q$-information density (3) relies on a closed-form geometric scaling via the single parameter $\alpha$. This obviates polynomial perturbations, providing a globally stable, non-oscillatory, and computationally efficient foundation for real-time evaluation in future 6G URLLC systems.Ultimately, this framework bridges FBL information theory and generalized statistical mechanics \cite{Tsallis2009, Suyari2026_JPA}. The scaling $1-q_n = O(n^{-1})$ is not merely a mathematical artifact to avoid negative probabilities, but a rigorous methodology that captures the intrinsic non-extensive geometry induced by finite-size constraints \cite{Suyari2026_PhysicaA_b}. Specifically, the exact dynamic tuning $\alpha = T / (3V^2)$ reveals a profound physical correspondence: the information-theoretic penalties at short blocklengths mathematically act as an effective non-extensive interaction among symbols. This formalizes the profound connection between the macroscopic limits in Shannon theory and the generalized thermodynamic formalism for finite-size complex systems, establishing a unified foundation where finite-size fluctuations and stability in power-law statistics are algebraically absorbed.

\section{Conclusion}
\label{sec:conclusion}
This paper revealed a structural limitation in applying the conventional Edgeworth expansion to FBL information theory: additive polynomial corrections yield unphysical negative probabilities in the deep-tail regime. To resolve this, we introduced the centralized $q$-information density. We proved that by scaling the deformation parameter dynamically as $1-q_n = O(n^{-1})$, the $q$-deformed framework exactly absorbs the dominant skewness, matching the third-order Cornish-Fisher bound. 

Furthermore, this structural deformation encapsulates higher-order asymptotic scales $O(n^{1-k/2})$ without introducing artificial polynomial roots. This exact structural absorption offers a globally stable, nonnegativity-preserving paradigm for large deviation evaluations in future finite-blocklength analysis, ensuring theoretically sound achievability bounds for ultra-reliable communication networks.

\section*{Acknowledgment}
This work was supported by JSPS KAKENHI Grant Number 26K14703.


\end{document}